\documentclass[12pt]{article}

\usepackage[margin=1in]{geometry}
\usepackage{CJKutf8, hyperref, amsmath, amssymb, amsthm, authblk, orcidlink}

\newtheorem{lemma}{Lemma}
\newtheorem{theorem}{Theorem}
\newtheorem{corollary}{Corollary}

\theoremstyle{definition}
\newtheorem{definition}{Definition}

\theoremstyle{remark}
\newtheorem*{remark}{Remark}

\DeclareMathOperator{\supp}{supp}
\DeclareMathOperator{\tr}{tr}
\DeclareMathOperator*{\E}{\mathbb E}

\usepackage[style=trad-unsrt, citestyle=numeric-comp, giveninits=true, doi=false, url=false, eprint=false, isbn=false]{biblatex}

\begin{document}

\title{Random product states at high temperature equilibrate exponentially well}

\begin{CJK}{UTF8}{gbsn}

\author{Yichen Huang (黄溢辰)\orcidlink{0000-0002-8496-9251}\thanks{yichenhuang@fas.harvard.edu}}

\affil{Department of Physics, Harvard University, Cambridge, Massachusetts 02138, USA}

\maketitle

\end{CJK}

\begin{abstract}

We prove that for all but a measure zero set of local Hamiltonians, starting from random product states at sufficiently high but finite temperature, with overwhelming probability expectation values of observables equilibrate such that at sufficiently long times, fluctuations around the stationary value are exponentially small in the system size.

\end{abstract}

\section{Introduction}

Equilibration is the process where the value of a dynamic property becomes almost time independent in the sense of staying very close to a particular value \cite{GE16}. If the stationary value matches the thermal value (property of the Gibbs state with the same energy), then thermalization occurs. Thus, equilibration is a prerequisite for thermalization. A thermalizing system equilibrates by definition, but not vice versa. The most famous example of ``equilibration without thermalization'' is arguably many-body localization.

Not all quantum many-body systems equilibrate. A simple counterexample is a non-interacting multi-spin system whose Hamiltonian is a sum of on-site terms. If the initial state is a product state, each spin evolves (rotates) independently and thus never reaches equilibrium. However and of course, we expect that almost all quantum many-body systems (with local interactions) equilibrate. This can be rigorously established from the laws of quantum mechanics. It follows from a line of research over a long period of time \cite{Tas98, Rei08, LPSW09, Sho11, FBC17, WGRE19, HH19, Hua21PP} that

\begin{theorem} [informal] \label{thm:1}
For all but a measure zero set of local Hamiltonians, starting from a random product state, with overwhelming probability expectation values of observables equilibrate such that at sufficiently long times, fluctuations around the stationary value are exponentially small in the system size.
\end{theorem}

\begin{remark}
This theorem has been extended to Floquet systems \cite{Hua24NJP}.
\end{remark}

In the thermodynamic limit, the energy density of a random product state approaches the mean energy density of the Hamiltonian (average of all eigenvalues divided by the system size) with overwhelming probability. Thus, Theorem \ref{thm:1} is a statement at infinite temperature.

The main contribution of this paper is to extend Theorem \ref{thm:1} to sufficiently high but finite temperature. Random product states at finite temperature can be defined by either
\begin{itemize}
\item a uniform distribution over a subset of product states whose energy corresponds to a given temperature (Definition \ref{def:mc}) or
\item properly defining the probability of each product state so that their mixture is a thermal state, as done in Ref.~\cite{BLMT24} (Definition \ref{def:can}).
\end{itemize}
We prove analogues of Theorem \ref{thm:1} for both definitions above.

\section{Preliminaries}

We use standard asymptotic notation. Let $f,g:\mathbb R^+\to\mathbb R^+$ be two functions. One writes $f(x)=O(g(x))$ if and only if there exist constants $C,x_0>0$ such that $f(x)\le Cg(x)$ for all $x>x_0$; $f(x)=\Omega(g(x))$ if and only if there exist constants $C,x_0>0$ such that $f(x)\ge Cg(x)$ for all $x>x_0$; $f(x)=\Theta(g(x))$ if and only if there exist constants $C_1,C_2,x_0>0$ such that $C_1g(x)\le f(x)\le C_2g(x)$ for all $x>x_0$; $f(x)=o(g(x))$ if and only if for any constant $C>0$ there exists a constant $x_0>0$ such that $f(x)<Cg(x)$ for all $x>x_0$.

\begin{definition} [non-degenerate gap]
The spectrum $\{e_j\}$ of a Hamiltonian has non-degenerate gaps if the differences $\{e_j-e_k\}_{j\neq k}$ are all distinct, i.e., for any $j\neq k$,
\begin{equation} \label{eq:ndg}
e_j-e_k=e_{j'}-e_{k'}\implies(j=j')~\textnormal{and}~(k=k').
\end{equation}
\end{definition}

Consider a system of $N$ qubits (spin-$1/2$'s) initialized in the state $|\psi\rangle$. Let $\{|j\rangle\}_{j=1}^{2^N}$ be a complete set of eigenstates of the Hamiltonian $H$.
\begin{definition} [effective dimension]
The effective dimension of $|\psi\rangle$ is defined as
\begin{equation}
\frac1{D^\textnormal{eff}_\psi}=\sum_{j=1}^{2^N}\big|\langle j|\psi\rangle\big|^4.
\end{equation}
\end{definition}

The time-averaged expectation value and fluctuation of an operator $B$ are
\begin{gather}
\bar B:=\lim_{T\to\infty}\frac1T\int_0^T\langle\psi|B(t)|\psi\rangle\,\mathrm dt,\quad B(t):=e^{iHt}Be^{-iHt},\\
\Delta B:=\lim_{T\to\infty}\frac1T\int_0^T\big|\langle\psi|B(t)|\psi\rangle-\bar B\big|^2\,\mathrm dt.
\end{gather}

Let $A$ be a subsystem of $L$ qubits and $\bar A$ be the rest of the system. Let
\begin{equation}
\psi_A(t)=\tr_{\bar A}(e^{-iHt}|\psi\rangle\langle\psi|e^{iHt}),\quad\bar\psi_A=\lim_{T\to\infty}\frac1T\int_0^T\psi_A(t)\,\mathrm dt
\end{equation}
be the reduced density matrix of subsystem $A$ at time $t$ and its time average.

\begin{theorem} [\cite{Rei08, LPSW09, Sho11}] \label{thm:equ}
If the spectrum of $H$ has non-degenerate gaps, then
\begin{equation}
\Delta B\le\|B\|^2/D^\textnormal{eff}_\psi,\quad\lim_{T\to\infty}\frac1T\int_0^T\|\psi_A(t)-\bar\psi_A\|_1\,\mathrm dt\le2^L/\sqrt{D^\textnormal{eff}_\psi}.
\end{equation}
\end{theorem}

\begin{theorem} [\cite{Hua21PP}] \label{thm:ndg}
In the space of local Hamiltonians, the non-degenerate gap condition (\ref{eq:ndg}) is satisfied almost everywhere.
\end{theorem}

Let $|x^\pm\rangle,|y^\pm\rangle,|z^\pm\rangle$ be the eigenvectors of the Pauli matrices
\begin{equation}
\sigma^x=\begin{pmatrix}0&1\\1&0\end{pmatrix},\quad
\sigma^y=\begin{pmatrix}0&-i\\i&0\end{pmatrix},\quad
\sigma^z=\begin{pmatrix}1&0\\0&-1\end{pmatrix}
\end{equation}
with eigenvalues $\pm1$, respectively.

\begin{definition}
Let $\mathcal E$ be the uniform distribution over the set
\begin{equation}
S:=\{|x^\pm\rangle,|y^\pm\rangle,|z^\pm\rangle\}^{\otimes N}.
\end{equation}
\end{definition}

\begin{theorem} \label{thm:eff}
\begin{equation}
\E_{|\psi\rangle\sim\mathcal E}\frac1{D^\textnormal{eff}_\psi}\le(2/3)^N,\quad\Pr_{|\psi\rangle\sim\mathcal E}(D^\textnormal{eff}_\psi=e^{\Omega(N)})=1-e^{-\Omega(N)}.
\end{equation}
\end{theorem}

This theorem can be proved in the same way as Lemma 5 in Ref.~\cite{HH19}. It is a special case of Lemma \ref{l:eff}, whose full proof is given below.

Theorem \ref{thm:1} follows from Theorems \ref{thm:equ}, \ref{thm:ndg}, \ref{thm:eff}.

\section{Results}

To extend Theorem \ref{thm:1} to finte temperature, it suffices to prove an analogue of Theorem \ref{thm:eff} for random product states at finite temperature.

Consider a constant-dimensional hypercubic lattice of $N$ sites, where each lattice site has a qubit. We expand the Hamiltonian
\begin{equation} \label{eq:H}
H=\sum_{P\in\mathcal P_N}c_PP
\end{equation}
in the Pauli basis, where $\mathcal P_N$ is the set of Pauli operators on $N$ qubits. Assume without loss of generality that $c_I=0$ so that $\tr H=0$. The support $\supp P$ of a Pauli operator is the set of qubits that $P$ acts non-trivially on. Suppose that $H$ is local and extensive in that
\begin{itemize}
\item $c_P=0$ if the diameter of $\supp P$ is greater than a certain constant $r$.
\item For any hypercubic region $R$ of side length $r$, there is at least one term in $H$ with $c_P\neq0$ such that $\supp P\subseteq R$.
\item $|c_P|=\Theta(1)$ if $c_P\neq0$.
\end{itemize}

Let
\begin{equation} \label{eq:gibbs}
\rho_\beta:=e^{-\beta H}/\tr(e^{-\beta H}),\quad E(\beta):=\tr(\rho_\beta H).
\end{equation}
be a thermal state and its energy at inverse temperature $\beta$.

\begin{definition} \label{def:mc}
Let $\mathcal E^\textnormal{mc}_\beta$ be the uniform distribution over the set
\begin{equation}
S_\beta:=\big\{|\psi\rangle\in S:\big|\langle\psi|H|\psi\rangle-E(\beta)\big|=o(N)\big\}.
\end{equation}
\end{definition}

\begin{theorem} \label{thm:main1}
There is a constant $C>0$ such that for any $|\beta|\le C$,
\begin{equation}
\E_{|\psi\rangle\sim\mathcal E^\textnormal{mc}_\beta}\frac1{D^\textnormal{eff}_\psi}=e^{-\Omega(N)},\quad\Pr_{|\psi\rangle\sim\mathcal E^\textnormal{mc}_\beta}(D^\textnormal{eff}_\psi=e^{\Omega(N)})=1-e^{-\Omega(N)}.
\end{equation}
\end{theorem}

Let $I_2$ be the identity matrix of order $2$ and
\begin{equation}
S^+:=\{|x^\pm\rangle,|y^\pm\rangle,|z^\pm\rangle,I_2/2\}^{\otimes N}.
\end{equation}

\begin{theorem} [Theorem 1.5 in Ref.~\cite{BLMT24}] \label{thm:un}
There is a constant $\beta_c>0$ such that for any $|\beta|\le\beta_c$, $\rho_\beta$ is separable:
\begin{equation}
\rho_\beta=\sum_{\rho\in S^+}p(\rho)\rho,
\end{equation}
where $p:S^+\to[0,1]$ is a probability distribution over $S^+$.
\end{theorem}

\begin{corollary} [Theorem 1.1 in Ref.~\cite{BLMT24}] \label{cor}
For any $|\beta|\le\beta_c$, $\rho_\beta$ can be written as a distribution over product states in $S$.
\end{corollary}

This corollary can be regarded as a definition of random product states at sufficiently high temperature \cite{PC24}.

We now explicitly derive Corollary \ref{cor} from Theorem \ref{thm:un}. Let $n(\rho)$ be the number of ``$I_2/2$'' in $\rho\in S^+$ so that $\langle\psi|\rho|\psi\rangle\le2^{-n(\rho)}$ for any $|\psi\rangle\in S$. Let
\begin{equation}
S_\rho:=\{|\psi\rangle\in S:\langle\psi|\rho|\psi\rangle=2^{-n(\rho)}\},
\end{equation}
i.e., $|\psi\rangle\in S_\rho$ if and only if $|\psi\rangle$ and $\rho$ are the same on qubits where $\rho$ is not $I_2/2$. Thus,
\begin{equation}
|S_\rho|=6^{n(\rho)},\quad\rho=\frac1{|S_\rho|}\sum_{|\psi\rangle\in S_\rho}|\psi\rangle\langle\psi|.
\end{equation}

\begin{definition} \label{def:can}
The proof of Theorem \ref{thm:un} is constructive. Let $p$ be the probability distribution constructed in Ref.~\cite{BLMT24}. We sample and obtain $\rho\in S^+$ with probability $p(\rho)$. Then, we sample uniformly from $S_\rho$. This results in a probability distribution $\mathcal E^\textnormal{can}_\beta$ over $S$ so that
\begin{equation}
\rho_\beta=\E_{|\psi\rangle\sim\mathcal E^\textnormal{can}_\beta}|\psi\rangle\langle\psi|.
\end{equation}
\end{definition}

\begin{remark}
$\mathcal E^\textnormal{can}_\beta$ is defined only for $|\beta|\le\beta_c$ where the construction of $p$ in Ref.~\cite{BLMT24} is provably valid.
\end{remark}

\begin{theorem} \label{thm:main2}
\begin{equation}
\E_{|\psi\rangle\sim\mathcal E^\textnormal{can}_\beta}\frac1{D^\textnormal{eff}_\psi}=e^{-\Omega(N)},\quad\Pr_{|\psi\rangle\sim\mathcal E^\textnormal{can}_\beta}(D^\textnormal{eff}_\psi=e^{\Omega(N)})=1-e^{-\Omega(N)}.
\end{equation}
\end{theorem}

\section{Proofs}

\subsection{Proof of Theorem \ref{thm:main1}}

\begin{lemma} \label{l:prodene}
\begin{equation}
\min_{|\psi\rangle\in S}\langle\psi|H|\psi\rangle=-\Omega(N).
\end{equation}
\end{lemma}

\begin{proof}
We construct hypercubic regions $R_1,R_2,\ldots,R_m$ of side length $r$ such that the distance between any two of them is greater than $r$. Since $r=O(1)$, we can have $m=\Omega(N)$.

For each $R_j$, there is at least one Pauli operator, denoted by $P_j$, with $c_{P_j}\neq0$ such that $\supp P_j\subseteq R_j$. If there is more than one such Pauli operator, $P_j$ with the smallest $|\supp P_j|$ is chosen. Let $|\phi_j\rangle$ be a pure product state of qubits in $\supp P_j$ such that $\langle\phi_j|c_{P_j}P_j|\phi_j\rangle=-|c_{P_j}|$.

Let $\rho\in S^+$ be the tensor product of $|\phi_1\rangle,|\phi_2\rangle,\ldots,|\phi_m\rangle$ and the maximally mixed state on the rest of the system so that $\tr(\rho P)\neq0$ only if
\begin{equation}
\supp P\subseteq\bigcup_{j=1}^m\supp P_j.
\end{equation}
Since the supports of $P_1,P_2,\ldots,P_m$ are pairwise separated by distance $r$,
\begin{equation}
\tr(\rho H)=-\sum_{j=1}^m|c_{P_j}|=-\Omega(N).
\end{equation}
Since there exists $|\psi\rangle\in S_\rho\subseteq S$ such that $\langle\psi|H|\psi\rangle\le\tr(\rho H)$, we complete the proof.
\end{proof}

The Hamiltonian (\ref{eq:H}) can be written as
\begin{equation}
H=\sum_iH_i.
\end{equation}
The sum is over lattice sites. Each term $H_i$ is traceless $\tr H_i=0$ and has bounded operator norm $\|H_i\|=O(1)$ and is supported in a constant-radius neighborhood of site $i$. Let $d(i,j)$ be the distance between two lattice sites $i,j$. At sufficiently high temperature, the thermal state (\ref{eq:gibbs}) has exponential decay of correlations \cite{Ara69, BCP22, KGK+14}. Therefore,
\begin{multline} \label{eq:Ebeta}
\frac{\mathrm dE(\beta)}{\mathrm d\beta}=\tr^2(\rho_\beta H)-\tr(\rho_\beta H^2)=\sum_{i,j}\tr(\rho_\beta H_i)\tr(\rho_\beta H_j)-\tr(\rho_\beta H_iH_j)=-\sum_{i,j}O(e^{-\Omega(d(i,j))})\\
=-O(N)\implies E(\beta)=-O(\beta N).
\end{multline}

\begin{lemma}
There is a constant $c>0$ such that
\begin{equation} \label{eq:cardSb}
|S_\beta|=\Omega(6^{N(1-c|\beta|)}),\quad\forall\beta\in[-1/c,1/c].
\end{equation}
\end{lemma}

\begin{proof}
Assume without loss of generality that $\beta\ge0$. If $c$ is sufficiently large, $\beta$ is sufficiently small so that (\ref{eq:Ebeta}) applies. We divide the lattice into two hyperrectangular regions $A$ and $\bar A$. By applying Lemma \ref{l:prodene} to region $A$, we can have a product state of $A$ such that the energy of $A$ is $<E(\beta)$, as long as $A$ has $O(\beta N)$ sites for a sufficiently large constant hidden in the Big-O notation. Therefore, there exists a product state of $A$ such that the energy of $A$ is $E(\beta)\pm O(1)$. Combining this product state of $A$ with almost any product state of $\bar A$ results in a product state whose total energy is $E(\beta)\pm o(N)$. Thus, $|S_\beta|=\Omega(6^{|\bar A|})$, where $|\bar A|$ is the number of sites in $\bar A$.
\end{proof}

We are ready to prove Theorem \ref{thm:main1}:
\begin{equation}
\E_{|\psi\rangle\sim\mathcal E^\textnormal{mc}_\beta}\frac1{D^\textnormal{eff}_\psi}=
\frac1{|S_\beta|}\sum_{|\psi\rangle\in S_\beta}\frac1{D^\textnormal{eff}_\psi}\le\frac1{|S_\beta|}\sum_{|\psi\rangle\in S}\frac1{D^\textnormal{eff}_\psi}=\frac{|S|}{|S_\beta|}\E_{|\psi\rangle\sim\mathcal E}\frac1{D^\textnormal{eff}_\psi}\le(2/3)^N\times O(6^{cN|\beta|})=e^{-\Omega(N)}
\end{equation}
if $|\beta|\le\frac1{5c}$.

\subsection{Proof of Theorem \ref{thm:main2}}

\begin{lemma} \label{l:Allen}
\begin{equation}
\sum_{\rho\in S^+:n(\rho)=\Omega(N)}p(\rho)=1-e^{-\Omega(N)}.
\end{equation}
\end{lemma}

\begin{proof}
A scheme for sampling $\rho\in S^+$ with probability $p(\rho)$ is given in Ref.~\cite{BLMT24}. The scheme has $N$ iterations. In each iteration, it samples the state of a new qubit from $\{|x^\pm\rangle,|y^\pm\rangle,|z^\pm\rangle,I_2/2\}$. It suffices to prove that the probability of sampling $I_2/2$ is always at least a positive constant. Then, the lemma follows from the Chernoff bound.

The procedure for sampling the state of a new qubit is Algorithm 6.3 in Ref.~\cite{BLMT24}. Please refer to that algorithm while reading the rest of the proof. In that algorithm, $j$ is the index of the qubit being sampled in the current step; $E,E_1,E_2$ are Pauli operators (not energy). In step 3, they sample from seven cases, each of which has a constant probability. The probability may be adjusted in Step 8, but the adjustment is only a positive constant factor. $E'$ is a Pauli operator computed from $E,E_1,E_2$. The formula for computing $E'$ is different in different cases. $E'_j\in\{I_2,\sigma^x,\sigma^y,\sigma^z\}$ is the restriction of $E'$ to qubit $j$ and determines which state to sample for qubit $j$.

Let $E_{0,j},E_{1,j},E_{2,j}$ be the restriction of $E,E_1,E_2$ to qubit $j$, respectively. $E'_j$ can be calculated from $E_{0,j},E_{1,j},E_{2,j}$ using the formulas in Step 3. By enumerating $(E_{0,j},E_{1,j},E_{2,j})\in\{I_2,\sigma^x,\sigma^y,\sigma^z\}^3$, we observe that among the seven cases, at least one of them leads to $E'_j=I_2$. Hence, the probability of sampling $I_2/2$ is always at least a positive constant.
\end{proof}

\begin{lemma} \label{l:eff}
For any $\rho\in S^+$,
\begin{equation}
\frac1{|S_\rho|}\sum_{|\psi\rangle\in S_\rho}\frac1{D^\textnormal{eff}_\psi}\le(2/3)^{n(\rho)}.
\end{equation}
\end{lemma}

\begin{proof}
Let $A$ be the set of qubits where $\rho$ is a pure state (not $I_2/2$) and $\bar A$ be the rest of the system. Then, states in $S_\rho$ can be parameterized as
\begin{equation}
|\phi\rangle_A\otimes\bigotimes_{j\in\bar A}|\phi_j\rangle,\quad|\phi_j\rangle\in\{|x^\pm\rangle,|y^\pm\rangle,|z^\pm\rangle\},
\end{equation}
where $|\phi\rangle_A$ is a fixed pure product state of subsystem $A$ determined by $\rho$.

It can be directly verified that
\begin{multline}
(|x^+\rangle\langle x^+|)^{\otimes2}+(|x^-\rangle\langle x^-|)^{\otimes2}+(|y^+\rangle\langle y^+|)^{\otimes2}+(|y^-\rangle\langle y^-|)^{\otimes2}+(|z^+\rangle\langle z^+|)^{\otimes2}+(|z^-\rangle\langle z^-|)^{\otimes2}\\
=2\Pi_\textnormal{sym},
\end{multline}
where $\Pi_\textnormal{sym}=\Pi_\textnormal{sym}^2$ is the projector onto the symmetric subspace of two qubits. Therefore,
\begin{align}
&\sum_{|\psi\rangle\in S_\rho}\big|\langle j|\psi\rangle\big|^4=\sum_{|\psi\rangle\in S_\rho}\langle j|^{\otimes2}(|\psi\rangle\langle\psi|)^{\otimes2}|j\rangle^{\otimes2}=\langle j|^{\otimes2}\sum_{|\psi\rangle\in S_\rho}(|\psi\rangle\langle\psi|)^{\otimes2}|j\rangle^{\otimes2}\nonumber\\
&=\langle j|^{\otimes2}\big((|\phi\rangle_A\langle\phi|_A)^{\otimes2}\otimes\bigotimes_{j\in\bar A}\sum_{|\phi_j\rangle}(|\phi_j\rangle\langle\phi_j|)^{\otimes2}\big)|j\rangle^{\otimes2}=\langle j|^{\otimes2}\big((|\phi\rangle_A\langle\phi|_A)^{\otimes2}\otimes(2\Pi_\textnormal{sym})^{\otimes n(\rho)}\big)|j\rangle^{\otimes2}\nonumber\\
&\le2^{n(\rho)}\langle j|^{\otimes2}\big((|\phi\rangle_A\langle\phi|_A)^{\otimes2}\otimes I_{\bar A}^{\otimes2}\big)|j\rangle^{\otimes2}\le2^{n(\rho)}\langle j|(|\phi\rangle_A\langle\phi|_A\otimes I_{\bar A})|j\rangle,
\end{align}
where $I_{\bar A}$ is the identity operator on $\bar A$. Summing over $j$,
\begin{equation}
\frac1{|S_\rho|}\sum_{|\psi\rangle\in S_\rho}\frac1{D^\textnormal{eff}_\psi}\le\frac1{6^{n(\rho)}}\sum_j2^{n(\rho)}\langle j|(|\phi\rangle_A\langle\phi|_A\otimes I_{\bar A})|j\rangle=(2/3)^{n(\rho)}.
\end{equation}
\end{proof}

Theorem \ref{thm:main2} follows directly from Lemmas \ref{l:Allen}, \ref{l:eff}.

\section*{Acknowledgments}

I would like to thank Aram W. Harrow for collaboration on closely related work \cite{HH19}. This work was supported by the Army Research Office (grant no.~W911NF-21-1-0262).

\printbibliography

\end{document}